\documentclass[11pt, helvetica]{article}

\usepackage{hyperref}

\def\thick#1{\hbox{\rlap{$#1$}\kern0.25pt\rlap{$#1$}\kern0.25pt$#1$}}

\def\smbalpha{\boldsymbol{{\scriptstyle{\alpha}}}}

\def\smbalpha{\widehat{\smbalpha}}

\def\hbar{\bar{ h}}

\def\mybox#1{\vskip1mm \begin{center}
        \hspace{.0\textwidth}\vbox{\hrule\hbox{\vrule\kern6pt
\parbox{.9\textwidth}{\kern6pt#1\vskip6pt}\kern6pt\vrule}\hrule}
        \end{center} \vskip-5mm}
\def\lboxit#1{\vbox{\hrule\hbox{\vrule\kern6pt
      \vbox{\kern6pt#1\vskip6pt}\kern6pt\vrule}\hrule}}

\def\thickboxit#1{\vbox{{\hrule height 1mm}\hbox{{\vrule width 1mm}\kern6pt
          \vbox{\kern6pt#1\kern6pt}\kern6pt{\vrule width 1mm}}
               {\hrule height 1mm}}}

\def\fat#1{\hbox{\rlap{$#1$}\kern0.25pt\rlap{$#1$}\kern0.25pt$#1$}}

\usepackage{authblk}  
\usepackage{bbm}
\usepackage[skins]{tcolorbox}
\usepackage{wrapfig,floatrow}
\usepackage{multicol}
\usepackage{verbatim}
\usepackage{times}
\usepackage{amsmath,amsthm,amssymb,mathrsfs}
\usepackage{color}
\usepackage{graphicx}
\usepackage{epsfig}
\usepackage{url}
\usepackage{xcolor}
\usepackage{xspace}
\usepackage{pgfgantt}
\usepackage{booktabs}
\usepackage{stfloats}
\usepackage{subcaption}
\usepackage{tcolorbox}
\usepackage[normalem]{ulem}
\usepackage[skins]{tcolorbox}
\usepackage{caption}
\usepackage[ruled]{algorithm2e}
\usepackage[T1]{fontenc}
\usepackage{enumitem} 

\definecolor{light-gray}{gray}{0.96}
 
 \newtheorem{theorem}{Theorem}
 \newtheorem{fact}{Fact}
 \newtheorem{lemma}{Lemma}
 \newtheorem{corollary}{Corollary}
 \newtheorem{claim}{Claim}

\newcommand{\defn}[1]{\textbf{\emph{#1}}}

\newcommand{\STB}{\textsc{Sawtooth Backoff}\xspace}
\newcommand{\BEB}{\textsc{Binary Exponential Backoff}\xspace}
\newcommand{\LLB}{\textsc{Log-Log Backoff}\xspace}
\newcommand{\FB}{\textsc{Fixed Backoff}\xspace}

\newcommand{\stb}{\textsc{STB}\xspace}
\newcommand{\beb}{\textsc{BEB}\xspace}
\newcommand{\llb}{\textsc{LLB}\xspace}
\newcommand{\fb}{\textsc{FB}\xspace}

\newcommand{\irvs}{i.r.v.{\footnotesize{s}}\xspace}

\def \[{ \left[ }
\def \]{\right] }

\def \Sbbm{\mathbbm{S}}

\definecolor{magenta4}{rgb}{0.5625,0,0.5625}
\definecolor{green4}{rgb}{0,0.5625,0}
\definecolor{orange4}{rgb}{0.98,0.31,0.09}

\newtheorem{property}{Property}

\newcommand{\maxwell}[1]{\ifcomments {\noindent \scriptsize  \textcolor{red} {Max: {#1}}} \fi{}}

\newif\ifcomments
\commentstrue

\begin{document}


\title{Singletons for Simpletons\\Revisiting Windowed Backoff with Chernoff Bounds}

\author{Qian M. Zhou}
\affil{Department of Mathematics and Statistics, Mississippi State University, MS, USA\\ \texttt{qz70@msstate.edu}}

\author{Alice Calvert\thanks{At the time this research was started, Alice Calvert was a high school student participating in an annual program between Mississippi State University and the Mississippi School for Mathematics and Science.}}
\affil{\texttt{aliceocalvert@protonmail.com}}

\author{Maxwell Young\thanks{This research is supported by the National Science Foundation grant CNS 1816076 and the National Institute of Justice grant 2018-75-CX-K002.}}
\affil{Department of Computer Science and Engineering, Mississippi State University, MS, USA\\ \texttt{myoung@cse.msstate.edu}}


\maketitle

\begin{abstract}
Backoff algorithms are used in many distributed systems where multiple devices contend for a shared resource. For the classic balls-into-bins problem, the number of singletons---those bins with a single ball---is important to the analysis of several backoff algorithms; however, existing analyses  employ advanced probabilistic tools. Here, we show that standard Chernoff bounds can be used instead, and the simplicity of this approach is illustrated by re-analyzing some well-known  backoff algorithms.
\end{abstract}


\section{Introduction}\label{sec:intro}

Backoff algorithms address the general problem of how to share a resource among multiple devices.  A popular application is IEEE 802.11 (WiFi) networks~\cite{802.11-standard,sun:backoff,kurose:computer}, where the resource is a wireless channel used by devices to send packets. Any single packet sent uninterrupted over the channel is likely to be received, but if the sending times of two or more packets overlap, communication often fails due to destructive interference at the receiver (i.e., a collision).  An important performance metric is the time required for all packets to be successfully sent, and this was originally analyzed for several well-known backoff algorithms by Bender et al.~\cite{BenderFaHe05}.

Here, we revisit several of these results, showing that a standard probabilistic tool---the Chernoff bound---is applicable, and illustrating how its use simplifies the analysis. A version of this work appeared at the $10^{\mbox{\tiny th}}$ {\it International Conference on Fun with Algorithms (FUN 2020)} under the {\defn{simplification of algorithms}} topic of interest.  

\medskip

\noindent{\bf Formal Model.} Time proceeds in discrete \defn{slots}, and each packet can be transmitted within a single slot. Starting from the first slot, a \defn{batch} of {\boldmath{$n$}} packets is ready to be transmitted on a shared channel.  Each packet can be viewed as originating from a different source device, and going forward we speak only of packets,  rather than devices.

For any fixed slot, if a single packet sends, then the packet \defn{succeeds}; however,  if two or more packets send, then all corresponding packets \defn{fail}. A packet that attempts to send in a slot learns whether it succeeded and, if so, the packet takes no further action; otherwise, the packet learns that it failed in that slot, and must try again at a later time. The number of slots required until all packets are successfully sent is known as the \defn{makespan}. \medskip

\noindent{\bf Background.} Binary exponential backoff (\beb) is a popular randomized algorithm for resource sharing. Originally introduced by Metcalf and Boggs~\cite{MetcalfeBo76} several decades ago, \beb is ubiquitous today, most notably in WiFi networks.  Execution occurs over disjoint, consecutive sets of slots called \defn{windows}. In every window, each packet that has not yet succeeded selects a single slot uniformly at random in which to send. If the packet succeeds, then it leaves the system; otherwise, the failed packet waits for the next window to begin and repeats this process.  For \beb, the $i^{\mbox{\tiny th}}$ window consists of $2^i$ slots, for $i\geq 0$. As we discuss later in Section~\ref{sec:describe-backoff}, other windowed backoff algorithms exist, where a different function of $i$ governs the corresponding window size.\smallskip


A natural question is the following: {\it{For a batch of $n$ packets and a given windowed backoff algorithm, what is the makespan?}} \smallskip

Interestingly, there is a close relationship between the makespan and the well-known balls-in-bins problem. In the latter, {\boldmath{$N$}} balls are dropped uniformly at random into {\boldmath{$B$}} bins. Associating the balls with packets, and the bins with the slots of a window, we are interested in the number of bins containing a single ball, which are sometimes referred to as \defn{singletons}~\cite{7282923}.

The makespan of backoff algorithms was first addressed  by Bender et al.~\cite{BenderFaHe05} who analyze several algorithms where windows monotonically increase in size. Despite their simple specification,  these  algorithms require a surprisingly intricate makespan analysis. In particular, obtaining concentration bounds on the number of slots (or bins) that contain a single packet (or ball)---which we will {\it also} refer to as singletons---is complicated by dependencies that rule out the immediate application of Chernoff bounds. This is unfortunate given that Chernoff bounds are often one of the first powerful probabilistic tools that researchers learn, and they are standard content in a randomized algorithms course.

The makespan results in Bender et al.~\cite{BenderFaHe05} are derived via delay sequences~\cite{Karlin:1986:PHE:12130.12146,Upfal:1984:ESP:828.1892}, which are arguably a less-common topic of instruction.  Alternative tools for handling dependencies include Poisson-based approaches by Mizenmacher~\cite{Mitzenmacher:1996:PTC:924815} and Mitzenmacher and Upfal~\cite{Mitzenmacher:2005:PCR:1076315}, and martingales~\cite{dubhashi:concentration}, but to the best of our knowledge, these have not been applied to the analysis of windowed backoff algorithms.\medskip




\noindent{\bf Our Goal.} We aim for a simpler way to derive the makespan of windowed backoff algorithms. This is accomplished in two parts: we show (1)  that a standard Chernoff bound can be used in the analysis, and (2) that this yields a straightforward approach to deriving the makespan. Our hope is that this work may improve accessibility to backoff algorithms for researchers learning about the area. \medskip

\noindent{\bf Benefits of Our Approach.} We note that claims of simplicity can be a matter of taste. There  is work involved in showing that Chernoff bounds can be applied to this problem. However,  in our opinion, this does not add much overhead to a rigorous introduction of Chernoff bounds that many researchers receive,  and it is arguably less complex than the alternatives.  For instance,  Dhubashi and Panconesi~\cite{dubhashi:concentration} derive Chernoff bounds almost immediately, starting on page $3$. In contrast, their discussion of negative dependence occurs in Chapter 3, while martingales and related tools for handling dependent random variables are deferred until Chapter 5.\footnote{An alternative route may involve negative dependence; this can be used to analyze the occupancy numbers of the balls-in-bins problem, although this is not what we need here. Alternatively, the Method of Bounded Differences (Corollary 5.2 in~\cite{dubhashi:concentration}), or more powerful tools, can be used to prove bounds that would allow us to analyze makespan, but the derivation is involved.}  

So, Chernoff bounds are often one of the first powerful probabilistic tools that researchers learn as they become involved in algorithms research. And what if we can deploy them to analyze singletons? Then, as we show here, the makespan analysis distills to proving the correctness of a ``guess'' regarding a recursive formula describing the number of packets remaining after each window, and that this guess has small error probability. Thus, once we can use Chernoff bounds, the approach for analyzing these backoff algorithms becomes simple and intuitive.

Finally, showing that another problem, especially one that has such a wide range of applications, succumbs to Chernoff bounds is  aesthetically satisfying.

\medskip

\subsection{Our Results} We show that Chernoff bounds can indeed be used as proposed above.  Our approach involves an argument that the indicator random variables for counting singletons satisfy the following property from~\cite{dubhashi:negative}:

\begin{property}\label{prop1}
Given a set of $n$ indicator random variables $\{X_1,\cdots,X_n\}$, for all subsets $\Sbbm \subset \{1,\cdots,n\}$ the following is true:
\begin{equation}\label{equ:neg-asso-defn}
Pr\left[\bigwedge\limits _{j\in \Sbbm} X_j=1 \right] \leq \prod_{j\in \Sbbm} Pr\[X_j=1\].
\end{equation} 
\end{property}

\noindent We prove the following:

\begin{theorem}\label{thm:neg-assoc-result}
Consider $N$ balls dropped uniformly at random into $B$ bins. Let $I_j=1$ if bin $j$ contains exactly $1$ ball, and $I_j=0$ otherwise, for $j=1,\cdots,B$. If $B\geq N+\sqrt{N}$ or $B\leq N-\sqrt{N}$, then $\{I_1,\cdots,I_B\}$ satisfy the Property~\ref{prop1}.
\end{theorem}

The proof of Property~\ref{prop1} permits the use of standard Chernoff bounds (defined in Theorem~\ref{thm:dhubashi-panconesi}). This implication is posed as an exercise by Dubhashi and Panconesi~\cite{dubhashi:concentration} (Problem 1.8), but no solution is provided~\cite{dubhashi:concentration}. We fill in this gap with our own argument. However, in order to avoid interrupting the flow of our main arguments for makespan analysis, we defer our discussion and proof until the appendix. \medskip

We then show how to use Chernoff bounds to obtain asymptotic makespan results for some of the algorithms  previously analyzed by Bender et al.~\cite{BenderFaHe05}: \BEB (\beb), \FB (\fb), and \LLB (\llb). Additionally, we re-analyze the asymptotically-optimal (non-monotonic) \STB (\stb) from~\cite{GreenbergL85,Gereb-GrausT92}.  These algorithms are specified in Section~\ref{sec:describe-backoff}, but our makespan results are stated below.  

\begin{theorem}
For a batch of $n$ packets, the following holds with probability at least $1-O(1/n)$:
\begin{itemize}
\item \fb has makespan at most $n\lg\lg n +  O(n)$. 
\item \beb has makespan at most $512n\lg n + O(n)$. 
\item \llb has makespan $O(n\lg\lg n/\lg\lg\lg n)$. \vspace{-3pt}
\item \stb has makespan $O(n)$.\vspace{-3pt}
\end{itemize}
\end{theorem}

We highlight that both of the cases in Theorem~\ref{thm:neg-assoc-result},  $B\leq N-\sqrt{N}$ and $B\geq N+\sqrt{N}$, are useful. Specifically,  the analysis for \beb, \fb, and \stb uses the first case, while \llb uses both.
 

\subsection{Related Work}

A preliminary version of this work appeared in the proceedings of the $10^{\mbox{\tiny th}}$ {\it International Conference on Fun with Algorithms (FUN 2020)} under the topic of interest {\defn{FUN with simplification of algorithms}}.  Here, we include our full proofs, further simplify parts of the analysis, and provide an expanded exposition. 

Several prior results address dependencies and their relevance to Chernoff bounds and load-balancing in various balls-in-bins scenarios.  In terms of analyzing backoff algorithms, the literature is vast. In both cases, we summarize  closely-related works. \medskip

\noindent{\bf Dependencies, Chernoff Bounds, \& Ball-in-Bins.} Backoff is closely-related to balls-and-bins problems~\cite{AzarBrKa99,ColeFrMa98,RichaMi01,Vocking03}, where balls and bins correspond to packets and slots, respectively. Balls-in-bins analysis often arises in problems of load balancing (for examples, see~\cite{BerenbrinkCzEg12,BerenbrinkCzSt06,BerenbrinkKhSa13}). Dubhashi and Ranjan~\cite{dubhashi:negative} prove that the occupancy numbers --- random variables $N_i$ denoting the number of balls that fall into bin $i$ --- are negatively associated. This result is used by Lenzen and Wattenhofer~\cite{lenzen:tight} use it to prove negative association for the random variables that correspond to {\underline {at most}} $k\geq 0$ balls. Czumaj and Stemann~\cite{czumaj:randomized} examine the maximum load in bins under an adaptive process where each ball is placed into a bin with minimum load of those sampled prior to placement. Finally, Dubhashi and Ranjan~\cite{dubhashi:negative} also show that Chernoff bounds remain applicable when the corresponding indicator random variables that are negatively associated. The same result is presented in Dubhashi and Panconesi~\cite{dubhashi:concentration}.\medskip


\noindent{\bf Backoff Algorithms.}  As summarized in Section~\ref{sec:intro}, for a batch of $n$ packets, a generic backoff algorithm executes over windows. For every window, each packet that has not yet succeeded selects a single slot uniformly at random in which to send. If the packet succeeds, then it leaves the system; otherwise, the failed packet waits for the next window to begin and repeats this process. Backoff algorithms differ in the way that they change their window size, and we further details in Section~\ref{sec:describe-backoff}. 

Many early results on backoff are given in the context of statistical queuing-theory  (see~\cite{hastad:analysis,GoodmanGrMaMa88,RaghavanUp99,GOLDBERG1999232,hastad:analysis,goldberg:contention}) where a common assumption is that packet-arrival times are Poisson distributed.

In contrast, for a batch of packets, backoff algorithms with  monotonically-increasing window sizes has been analyzed in~\cite{BenderFaHe05}, and with packets of different sizes in~\cite{bender:heterogeneous}. A windowed, non-monotonic backoff algorithm that is asymptotically optimal for a batch of packets is provided in~\cite{Gereb-GrausT92,GreenbergL85,FernandezAntaMoMu13}.  

A related problem is  {\it contention resolution}, which addresses the time until the first packet succeeds~\cite{willard:loglog,nakano2002,FinemanNW16,fineman:contention}. This has close ties to the well-known problem of leader election (for examples, see~\cite{ChangKPWZ17,Chang:2018:ECB:3212734.3212774}).

Several results examine the {\it dynamic} case where packets arrive over time as scheduled in a worst-case fashion~\cite{doi:10.1137/140982763,DeMarco:2017:ASC:3087801.3087831,Bender:2016:CRL:2897518.2897655,DBLP:conf/spaa/AgrawalBFGY20}. A similar problem is that of
{\it wake-up}~\cite{chlebus:better,chlebus:wakeup,chrobak:wakeup,Chlebus:2016:SWM:2882263.2882514,DEMARCO20171,Jurdzinski:2015:CSM:2767386.2767439}, which addresses how long it takes for a single transmission to succeed when packets arrive under the dynamic scenario. 

Finally, several results address the case where the shared communication channel is unavailable at due to malicious interference~\cite{awerbuch:jamming,richa:jamming2,richa:jamming3,richa:jamming4,ogierman:competitive,Anantharamu2011,BenderFiGi16,DBLP:journals/jacm/BenderFGY19}.


\section{Analysis for Property~\ref{prop1} }

We present our results on Property~\ref{prop1}. Since we believe this result may be useful outside of backoff, our presentation in this section is given in terms of the well-known balls-in-bins terminology, where we have {\boldmath{$N$}} balls that are dropped uniformly at random into {\boldmath{$B$}} bins.


\subsection{Preliminaries}\label{sec:prelim}

Throughout, we often employ the following inequalities (see Lemma 3.3 in~\cite{richa:jamming4}), and we will refer to the left-hand side (LHS) or right-hand side (RHS) when doing so.
\begin{fact}\label{fact:taylor}
For any $0<x<1$,   $e^{-x/(1-x)} \leq 1 - x \leq e^{-x}$.
\end{fact}

\noindent Knowing that indicator random variables (\irvs) satisfy Property~\ref{prop1}  is useful since the following Chernoff bounds can then be applied.
\begin{theorem}\label{thm:dhubashi-panconesi}
(Dubhashi and Panconesi~\cite{dubhashi:concentration})\footnote{This is stated in Problem 1.8 in~\cite{dubhashi:concentration}, and we present a proof in~\ref{app:property1} .}   Let $X=\sum_{i} X_i$ where $X_1, ..., X_m$ are \irvs that satisfy Property~\ref{prop1} . For $0< \epsilon < 1$, the following holds:
\begin{eqnarray}
Pr[X > (1+\epsilon)E[X] ] &\leq& \exp\left( -\frac{\epsilon^2}{3}E[X] \right)\label{eqn:chernoff-upper}\\
Pr[X  < (1-\epsilon)E[X] ] &\leq& \exp\left( -\frac{\epsilon^2}{2}E[X] \right)\label{eqn:chernoff-lower}
\end{eqnarray}
\end{theorem}

\noindent We are  interested in the \irvs $I_j$, where:
\begin{equation*}
  I_j=\begin{cases}
    1, & \text{if bin $j$ contains exactly $1$ ball}.\\
    0, & \text{otherwise}.
  \end{cases}
\end{equation*}
\noindent 
\noindent Unfortunately, there are cases where the $I_j$s fail to satisfy Property~\ref{prop1}. For example, consider $N=2$ balls and $B=2$ bins. Then, $Pr(I_1=1) = Pr(I_2=1)=1/2$, so $Pr(I_1=1) \cdot Pr(I_2=1)=1/4$,  but $Pr( I_1 =1 \wedge I_2=1) = 1/2$.

A naive approach (although, we have not seen it in the literature) is to leverage the result in~\cite{lenzen:tight}, that the variables used to count the number of bins with {\underline{at most}} $k$ balls are negatively associated. We may bound the number of bins that have at most $1$ ball, and the number of bins that have (at most) $0$ balls, and then take the difference. However, this is a cumbersome approach, and our result is more direct. 

Returning briefly to the context of packets and time slots, another approach is to consider a subtly-different algorithm where a packet sends with probability $1/w$ in each slot of a window with $w$ slots, rather than selecting uniformly at random a single slot to send in. However, as Bender et al.~\cite{BenderFaHe05} point out, when $n$ is within a constant factor of the window size, there is a constant probability that the packet will not send in {\it any} slot. Consequently, the number of windows required for all packets to succeed increases by a $\log n$-factor, whereas only $O(\log\log n)$ windows are required under the model used here.



\subsection{Property~\ref{prop1} and Bounding Singletons}

To prove Theorem~\ref{thm:neg-assoc-result}, we establish the following Lemma \ref{thm:cond-prob}. For $j=1,\cdots,B-1$, define:
 $$\mathcal P_j = Pr\left[I_{j+1}=1 \mid I_1=1,\cdots, I_j=1\right]$$
 \noindent which is the conditional probability that bin $j+1$ contains exactly 1 ball given each of the bins $\{1,\cdots,j\}$ contains exactly 1 ball. Note that $Pr[I_j=1]$ is same for any $j=1,\cdots,B$, and let:
 \begin{equation}\label{equ:P0}
\mathcal P_0 \triangleq Pr[I_j=1] = N\left(\frac{1}{B}\right)\left(1-\frac{1}{B}\right)^{N-1}.
\end{equation}

\begin{lemma}\label{thm:cond-prob}
 If $B\geq N+\sqrt{N}$ or $B\leq N-\sqrt{N}$, the conditional probability $\mathcal P_j$ is a monotonically non-increasing function of $j$, i.e., $\mathcal P_j \geq \mathcal P_{j+1}$, for $j=0,\cdots,B-2$.
\end{lemma}
\begin{proof}
First, for $j=1,\cdots,\min\{B,N\}-1$, the conditional probability can be expressed as
\begin{equation}\label{equ:cond-prob}
\mathcal P_j = (N-j)\left(\frac{1}{B-j}\right)\left(1-\frac{1}{B-j}\right)^{N-j-1}.
\end{equation}
Note that $\mathcal P_0$ in (\ref{equ:P0}) is equal to (\ref{equ:cond-prob}) with $j=0$.

For $B\geq N+\sqrt{N}$, we note that beyond the range $j=1, ..., \min\{B,N\}-1$, it must be that $\mathcal P_j=0$. In other words, $\mathcal P_j=0$ for $j=N,N+1,\cdots,B-1$ since all balls have already been placed. Thus, we need to prove $\mathcal P_j \geq \mathcal P_{j+1}$, for $j=0,\cdots,N-2$. 

On the other hand, if $B\leq N-\sqrt{N}$, we need to prove $\mathcal P_j \geq \mathcal P_{j+1}$, for $j=0,\cdots,B-2$. Thus, this lemma is equivalent to prove  if $B\geq N+\sqrt{N}$ or $B\leq N-\sqrt{N}$, the ratio $\mathcal P_j/\mathcal P_{j+1}\geq 1$, for $j=0,\cdots,\min\{B,N\}-2$. 

\noindent Using the Equation (\ref{equ:cond-prob}), the ratio can be expressed as:
{\begin{align*}
\frac{\mathcal P_j}{\mathcal P_{j+1}}& =   \frac{(N-j)\left(\frac{1}{B-j}\right)\left(1-\frac{1}{B-j}\right)^{N-j-1}}{(N-j-1)\left(\frac{1}{B-j-1}\right)\left(1-\frac{1}{B-j-1}\right)^{N-j-2}} = \frac{\left(1+\frac{1}{(B-j)(B-j-2)}\right)^{N-j-1}}{\frac{(N-j-1)(B-j)}{(N-j)(B-j-2)}}.
\end{align*} }

\noindent Let $a = N-j$, then $2 \leq a \leq N$; and let $y = B-N$. Thus, the ratio becomes $\frac{\mathcal P_j}{\mathcal P_{j+1}} = \frac{\left[1+\frac{1}{(a+y)(a+y-2)}\right]^{a-1}}{\frac{(a-1)(a+y)}{a(a+y-2)}}$.
%
By the Binomial theorem, we have:
{\small \begin{eqnarray*}
\left[1+\frac{1}{(a+y)(a+y-2)}\right]^{a-1} \hspace{-22pt} & =& \hspace{-10pt}1 \hspace{-2pt}+\hspace{-2pt} \frac{a-1}{(a+y)(a+y-2)} \hspace{-2pt}+\hspace{-2pt}\sum_{k=2}^{a-1}\binom{a-1}{k}\hspace{-4pt} \left[\frac{1}{(a+y)(a+y-2)}\right]^{k}\hspace{-4pt}.
\end{eqnarray*}
}
Thus, the ratio can be written as:
{\small \begin{eqnarray}
\frac{\mathcal P_j}{\mathcal P_{j+1}} &=& 1 + \frac{y^2-a}{(a+y)^2(a-1)} + \frac{\sum_{k=2}^{a-1}\binom{a-1}{k} \left[\frac{1}{(a+y)(a+y-2)}\right]^{k}}{\frac{(a-1)(a+y)}{a(a+y-2)}}.\label{equ:ratio}
\end{eqnarray}
}
\noindent{}Note that because $0\leq j \leq \min\{B,N\}-2$, then $a+y=B-j\geq 2$. Thus, the third term in (\ref{equ:ratio}) is always non-negative. If $y=B-N\geq \sqrt{N}$ or $y\leq -\sqrt{N}$, then $y^2\geq N \geq a$ for any $2 \leq a \leq N$. Consequently, the ratio $\mathcal P_j /\mathcal P_{j+1}\geq 1$.
\end{proof}


\noindent{}We can now give our main argument:

\begin{proof}[Proof of Theorem~\ref{thm:neg-assoc-result}] Let $s$ denote the size of the subset $\mathbbm S \subset \{1,\cdots,B\}$, i.e. the number of bins in $\mathbbm S$. First, note that if $B\geq N+\sqrt{N}$, when $s>N$ (i.e., more bins than balls), the probability on the left hand side (LHS) of (\ref{equ:neg-asso-defn}) is 0, thus, the inequality (\ref{equ:neg-asso-defn}) holds. In addition, shown above $Pr[I_j=1] = \mathcal P_0$ for any $j=1,\cdots,B$. Thus, the right hand side of  (\ref{equ:neg-asso-defn}) becomes $\mathcal P_0^s$. Thus, we need to prove for any subset, denoted as $\mathbbm S = \{j_1,\cdots,j_s\}$ with $1 \leq s \leq \min\{B,N\}$ 
\begin{eqnarray*}
Pr\left[\bigwedge\limits _{k=1}^s I_{j_k}=1 \right]  \leq \mathcal P_0^s.
\end{eqnarray*}
\noindent The LHS can be written as:
\begin{eqnarray*}
 &=& Pr\left[I_{j_s}=1 \mid \bigwedge\limits _{k=1}^{s-1} I_{j_k}=1\right]Pr\left[\bigwedge\limits _{k=1}^{s-1} I_{j_k}=1\right] \\
&& \hspace{10pt}= \mathcal P_{s-1}Pr\left[\bigwedge\limits _{k=1}^{s-1} I_{j_k}=1\right] 
  \end{eqnarray*}
 \begin{eqnarray*}
&\hspace{-15pt} =& \hspace{-10pt} \mathcal P_{s-1} Pr\left[I_{j_{s-1}}=1 \mid \bigwedge\limits _{k=1}^{s-2} I_{j_k}=1\right] Pr\left[\bigwedge\limits _{k=1}^{s-2} I_{j_k}=1\right]\\
 && \hspace{10pt} = \mathcal P_{s-1} \mathcal P_{s-2} Pr\left[\bigwedge\limits _{k=1}^{s-2} I_{j_k}=1\right]\\
& \hspace{-15pt}&    \hspace{100pt} \vdots\\
&\hspace{-15pt} = & \hspace{-10pt}\mathcal P_{s-1} \mathcal P_{s-2} \cdots \mathcal P_0
\end{eqnarray*}
Lemma \ref{thm:cond-prob} shows that if $B\geq N+\sqrt{N}$ or $B\leq N-\sqrt{N}$, $\mathcal P_j$ is a decreasing function of $j=0,\cdots,B-1$. Consequently, $\mathcal P_0 \geq \mathcal P_j$, for $j=1,\cdots,B-1$. Thus:
$$
Pr\left[\bigwedge\limits _{k=1}^s I_{j_k}=1 \right]  \leq \mathcal P_0^s,
$$
\noindent and so the bound in Equation (\ref{equ:neg-asso-defn}) holds.
\end{proof}


The standard Chernoff bounds of Theorem~\ref{thm:dhubashi-panconesi} now apply, and we use them obtain bounds on the number of singletons. For ease of presentation, we occasionally use $\exp(x)$ to denote $e^{x}$.
 
\begin{lemma}\label{lem:concentration}
For $N$ balls that are dropped into $B$ bins where $B\geq N+\sqrt{N}$ or $B\leq N-\sqrt{N}$, the following is true for any  $0 < \epsilon < 1$.
\begin{itemize}[leftmargin=10pt]
\item \hspace{-2pt}The number of singletons is at least $\frac{(1-\epsilon)N}{e^{N/(B-1)}}$ with probability at least $1-$ $e^{\frac{-\epsilon^2 N }{2\exp(N/(B-1))}}$.
\item \hspace{-2pt}The number of singletons is at most $\frac{(1+\epsilon)N}{e^{(N-1)/B}}$  with probability at least $1- e^{\frac{-\epsilon^2 N }{3\exp(N/(B-1))}}$.
\end{itemize}
\end{lemma}
\begin{proof}
We begin by calculating the expected number of singletons. Let $I_i$ be an indicator random variable such that $I_i=1$ if bin $i$ contains a single ball; otherwise, $I_i=0$. Note that:
\begin{eqnarray}
Pr(I_i=1) & = & \binom{N}{1}\left(\frac{1}{B}\right) \left( 1- \frac{1}{B}\right)^{N-1}\nonumber\\
 & \geq & \binom{N}{1}\left(\frac{1}{B}\right) \left( 1- \frac{1}{B}\right)^{N}\nonumber\\ 
&\geq& \frac{N}{Be^{(N/(B-1))}}\label{eqn:problower}
\end{eqnarray}
\noindent where the last line follows from the LHS of Fact~\ref{fact:taylor}. Let $I = \sum_{i=1}^{B} I_i$ be the number of singletons. We have:
\begin{eqnarray*}
E[I] & = & \sum_{i=1}^{B} E[I_i]  \mbox{~~~by linearity of expectation}\\
& \geq & \frac{N}{e^{(N/(B-1))}}  \mbox{~~~by Equation~(\ref{eqn:problower})} 
\end{eqnarray*}
\noindent  Next, we derive a concentration result around this expected value. Since $B\geq N+\sqrt{N}$ or $B\leq N-\sqrt{N}$, Theorem~\ref{thm:neg-assoc-result} guarantees that the $I_i$s are negatively associated, and we may apply the Chernoff bound in Equation~\ref{eqn:chernoff-lower} to obtain:
\begin{eqnarray*}
Pr\left(I < (1-\epsilon) \frac{N}{e^{(N/(B-1))}}  \right) \hspace{-5pt}&\leq& \hspace{-5pt}\exp\left( -\frac{\epsilon^2 N }{2e^{(N/(B-1))}} \right)\\
\end{eqnarray*}
\noindent which completes the lower-bound argument. 

For the upper bound, we have:
\begin{eqnarray}
Pr(I_i=1) & = & \binom{N}{1}\left(\frac{1}{B}\right) \left( 1- \frac{1}{B}\right)^{N-1}\nonumber\\
 & \leq& \left(\frac{N}{B}\right) e^{-(N-1)/B}\nonumber\\ 
\end{eqnarray} 
\noindent where the inequality follows from the RHS of Fact~\ref{fact:taylor}. The expected value calculation and the application of Chernoff bounds follow by a nearly identical argument to the one presented above.
\end{proof}

\noindent We note that our makespan analysis makes use of the lower bound given in Lemma~\ref{lem:concentration}, since we are concerned with ensuring singletons occur. Throughout, we assume this lower bound is being used when Lemma~\ref{lem:concentration} is invoked. The upper bound is given for completeness and may be useful for proving lower bounds on makespan, as mentioned in Section~\ref{sec:conclusion}.


\section{Bounding Remaining Packets}\label{sec:bounding-remaining-packets}

In this section, we derive tools for bounding the number of packets that remain as we progress from one window to the next. 

All of our results hold for sufficiently large $n>0$.  Let {\boldmath{$w_i$}} denote the number of slots in window $i\geq 0$. Let {\boldmath{$m_i$}} be the number of packets at the start of window $i\geq 0$.

We index windows starting from $0$, but this does not necessarily correspond to the initial window executed by a backoff algorithm. Rather, in our analysis, window $0$ corresponds to the first window where packets start to succeed in large numbers; this is different for different backoff algorithms.

For example, \beb's initial window consists of a single slot, and does not play an important role in the makespan analysis. Instead, we apply Chernoff bounds once the window size is at least $n+\sqrt{n}$, and this corresponds to window $0$. In contrast, for \fb, the first window (indeed, {\it each} window) has size $\Theta(n)$, and window $0$ is indeed this first window for our analysis. This indexing is useful for our inductive arguments presented in Section~\ref{sec:induction}.\vspace{-5pt}

\subsection{Analysis}\vspace{-3pt}

Our method for upper-bounding the makespan operates in three stages. First, we apply an inductive argument---employing Case 1 in Corollary~\ref{cor:remaining} below---to cut down the number of packets from $n$ to less than $n^{7/10}$.  Second, Case 2 of Corollary~\ref{cor:remaining} is used to reduce the number of remaining packets  to $O(n^{2/5})$. Third, we hit the remaining packets with a constant number of calls to Lemma~\ref{lem:last-window}; this is the essence of Lemma~\ref{lem:end}.\medskip

\noindent{\bf Intuition for Our Approach.} There are a couple things worth noting. To begin, why not carry the inductive argument further to reduce the number of packets to $O(n^{2/5})$ directly (i.e., skip the second step above)? Informally, our later inductive arguments show that $m_{i+1}$ is roughly at most $n/2^{2^{i}}$, and so $i\approx \lg\lg(n)$ windows should be sufficient. However,  $\lg\lg(n)$ is not necessarily an integer and we may need to take its floor. Given the double exponential, taking the floor (subtracting $1$) results in $m_{i+1} \geq \sqrt{n}$. Therefore, the equivalent of our second step will still be required. Our choice of $n^{7/10}$ is not the tightest, but it is chosen for simpicity.

The second threshold of $O(n^{2/5})$ is also not completely arbitrary.  In the (common) case where $w_0 \geq n+\sqrt{n}$, note that we require $O(n^{1/2-\delta})$ packets remaining, for some constant $\delta>0$, in order to get a useful bound from Lemma~\ref{lem:last-window}. It is possible that after the inductive argument, that this is already satisfied; however, if not, then Case 2 of Corollary~\ref{cor:remaining} enforces this. Again, $O(n^{2/5})$ is chosen for ease of presentation; there is some slack.

\begin{corollary}\label{cor:remaining}
For $w_i\geq n+\sqrt{n}$, the following is true with probability at least $1-1/n^2$:\vspace{-5pt}
\begin{itemize}[leftmargin=4mm]
\item{\it Case 1.} If $m_i\geq n^{7/10}$, then $m_{i+1} < \frac{(5/4)m_i^2}{n}$.\vspace{-5pt}
\item{\it Case 2.} If\hspace{-1pt} $n^{2/5}\hspace{-2pt} \leq\hspace{-2pt}  m_i\hspace{-2pt}  < \hspace{-2pt} n^{7/10}$, \hspace{-2pt} then $m_{i+1}\hspace{-2pt} =\hspace{-2pt} O(n^{2/5})$. \vspace{-5pt}
\end{itemize}
\end{corollary}
\begin{proof}
For Case 1, we apply the first result of Lemma~\ref{lem:concentration} with $\epsilon =\frac{ \sqrt{4e\ln n}}{n^{1/3}}$, which implies with probability at least $1 - \exp(-\frac{4e\ln n}{n^{2/3}}\frac{n^{7/10}}{2}) \geq 1 - \exp(-2\ln n)\geq 1-1/n^2$: \vspace{-10pt}

\begin{eqnarray} 
m_{i+1} &\leq & m_i - \frac{(1-\epsilon)m_i}{e^{m_i/(w_i-1)}}\nonumber \\
& \leq & m_i\left(1 - \frac{1}{e^{m_i/(w_i-1)}} + \epsilon \right) \nonumber \\
& \leq & m_i\left(\frac{m_i}{w_i-1} + \epsilon \right)\nonumber  \mbox{~~by RHS of Fact~\ref{fact:taylor}} \\
& \leq & \frac{m_i^2}{n} +  m_i\epsilon \mbox{~~~since $w_i\geq n+\sqrt{n}$}\nonumber\\
& \leq & \frac{m_i^2}{n} + \left( \frac{m_i}{n^{1/3}}\right)\sqrt{4e\ln n}\label{eqn:caseblahh}\\
& <  & \frac{(5/4)m_i^2}{n} \mbox{~~~since $m_i\geq n^{7/10}$}\nonumber
\end{eqnarray}
\noindent where $5/4$ is chosen for ease of presentation.

For Case 2, the two terms in Equation~\ref{eqn:caseblahh} are $n^{2/5}$ and $O(n^{(7/10) - (1/3)}\sqrt{\ln n})$, respectively, for the any  $n^{2/5}\leq m_i  \leq n^{7/10}$; thus, $m_{i+1} = O(n^{2/5})$, as claimed.
\end{proof}

\noindent The following lemma is useful for achieving a with-high-probability guarantee when the number of balls is small relative to the number of bins.\smallskip

\begin{lemma}\label{lem:last-window}
Assume $w_i > 2m_i$. With probability at least $1 - \frac{m_i^2}{w_i} $, all packets succeed in window $i$.
\end{lemma}
\begin{proof}
Consider placements of packets in the window that yield at most one packet per slot. Note that once a packet is placed in a slot, there is one less slot available for each remaining packet yet to be placed. Therefore, there are $w_i(w_i-1)\cdots(w_i-m_i+1)$ such placements.  

Since there are $w_i^{m_i}$ ways to place $m_i$ packets in $w_i$ slots, it follows that the probability that each of the $m_i$ packets chooses a different slot is:
\begin{eqnarray*}
\frac{w_i(w_i-1)\cdots (w_i-m_i+1)}{w_i^{m_i}}.
\end{eqnarray*}
\noindent We can lower bound this probability:
\begin{eqnarray*}
 &=& \frac{w_i^{m_i} (1-1/w_i) \cdots (1-(m_i-1)/w_i)}{w_i^{m_i}} \\ 
& \geq & e^{-\sum_{j=1}^{m_i-1} \frac{j}{w_i-j} } \mbox{~~~by LHS of Fact~\ref{fact:taylor}}\\
&\geq & e^{-\sum_{j=1}^{m_i-1} \frac{2j}{w_i} } \mbox{~~~since $w_i > 2m_i > 2j$ which}\\
&&\mbox{\hspace{57pt} leads to $\frac{j}{w_i-j} < \frac{2j}{w_i}$}\\
& = &e^{-(1/w_i)(m_i-1)m_i} \mbox{~~~by  sum of natural numbers}\\
& \geq & 1 - \frac{m_i^2}{w_i} + \frac{m_i}{w_i} \mbox{~~~by RHS of Fact~\ref{fact:taylor}}\\
& > & 1 - \frac{m_i^2}{w_i}
\end{eqnarray*}
as claimed.
\end{proof}


\begin{lemma}\label{lem:end} 
Assume a batch of $m_i< n^{7/10}$ packets that execute over a window of size $w_i$, where $w_i \geq n + \sqrt{n}$ for all $i$. Then, with probability at least $1-O(1/n)$, any monotonic backoff algorithm requires at most $6$ additional windows for all remaining packets to succeed.
\end{lemma}
\begin{proof}
If $m_i \geq n^{2/5}$, then Case 2 of Corollary~\ref{cor:remaining} implies $m_{i+1}=O(n^{2/5})$; else, we do not need to invoke this case. By Lemma~\ref{lem:last-window}, the probability that any packets remain by the end of window $i+1$ is $O(n^{4/5}/n) = O(1/n^{1/5})$; refer to this as the probability of failure. Subsequent windows increase in size monotonically, while the number of remaining packets  decreases monotonically. Therefore, the probability of failure is  $O(1/n^{1/5})$ in any subsequent window, and the probability of failing over all of the next $5$ windows is less than $O(1/n)$. It follows that at most $6$ windows are needed for all packets to succeed.
\end{proof}


\section{Inductive Arguments}\label{sec:induction}

We present two inductive arguments for establishing upper bounds on $m_{i}$. Later in Section~\ref{sec:reanalyze}, these results are leveraged in our makespan analysis, and extracting them here allows us to modularize our presentation. Lemma~\ref{lem:induction-new} applies to \fb, \beb, and \llb, while Lemma~\ref{lem:induction-new-stb} applies to \stb. We highlight that a single inductive argument would suffice for all algorithms --- allowing for a simpler presentation --- if we only cared about asymptotic makespan. However, in the case of \fb we wish to obtain a tight bound on the first-order term, which is one of the contributions in~\cite{BenderFaHe05}. 

In the following lemmas, it is worth noting that we specify $m_0 \leq n$ instead of $m_0 = n$,  since some packets may have succeeded prior to window $0$; recall, this is the window where a large number of packets are expected to succeed. That said, the following inductive arguments apply so long as $m_0 \geq n^{7/10}$.

We also highlight the error probability present in each of the following lemmas; that is, the decrease in the number of remaining packets specified in these lemmas fails to hold with some probability. However, this probability is polynomially small in $n$, and this is what allows us to bound the total error over the $O(\log\log n)$ applications of these lemmas in  Section~\ref{sec:final-analysis}.

\begin{lemma}\label{lem:induction-new}
Consider a batch of $m_0\leq n$ packets that execute over windows $w_i \geq m_0 + \sqrt{m_0}$  for all $i\geq 0$. If $m_i \geq n^{7/10}$, then $m_{i+1} \leq (4/5)\frac{m_0}{2^{2^i\lg(5/4)}}$ with error probability at most $(i+1)/n^2$.
\end{lemma}
\begin{proof}
We argue by induction on $i\geq 0$.
\medskip

\noindent{\bf{Base Case}.} Let $i=0$.  Using Lemma~\ref{lem:concentration}: 
\begin{eqnarray*}
m_1 & \leq & m_0 - \frac{(1-\epsilon)m_0}{e^{m_0/(w_0-1)}}\nonumber \\
& \leq & m_0 \left( 1 - \frac{1}{e^{m_0/(w_0-1)}} + \epsilon \right) \\
  & \leq & m_0 \left( 1 - \frac{1}{e} + \epsilon \right) \\
  & \leq & (16/25)m_0
\end{eqnarray*}
where the last line follows by setting $\epsilon = \frac{\sqrt{4e\ln n}}{n^{1/3}}$, and assuming $n$ is sufficiently large; this gives an error probability of at most $1/n^2$ . The base case is satisfied since $(4/5)\frac{m_0}{2^{2^i \lg(5/4)}} = (4/5)\frac{m_0}{2^{\lg(5/4)}} = (16/25)m_0$.\medskip

\noindent{\bf{Induction Hypothesis (IH)}.} For  $i\geq 1$,  assume $m_{i} \leq (4/5)\frac{m_0}{2^{2^{i-1}\lg(5/4)}}$ with error probability at most $i/n^2$.
\medskip

\noindent{\bf{Induction Step}.} For window $i\geq 1$, we wish to show that $m_{i+1} \leq (4/5)\frac{m_0}{2^{2^i\lg(5/4)}}$ with an error bound of $(i+1)/n^2$. Addressing the number of packets, we have:
\begin{eqnarray*}
m_{i+1}\hspace{-5pt}&\leq&\hspace{-5pt} \frac{(5/4)m_i^2}{w_i} \nonumber \\
\hspace{-5pt}&\leq&\hspace{-5pt}   \left(\frac{4\,m_0}{5\cdot 2^{2^{i-1}\lg(5/4)}}\right)^2 \left(\frac{5}{4w_i} \right)  \\
\hspace{-5pt}&\leq&\hspace{-5pt}   \left(\frac{4m_0}{5\cdot 2^{2^{i}\lg(5/4)}}\right)\left(\frac{m_0}{w_i} \right)  \\
\hspace{-5pt}& < & \hspace{-5pt}  \left(\frac{4m_0}{5\cdot 2^{2^{i}\lg(5/4)}}\right) \mbox{since $w_i>m_0$} 
 \end{eqnarray*}
\noindent The first line follows from Case 1 of Corollary~\ref{cor:remaining}, which we may invoke since $w_i \geq m_0 + \sqrt{m_0}$  for all $i\geq 0$, and $m_i \geq n^{7/10}$ by assumption. This yields an error of at most $1/n^2$, and so the total error is at most $i/n^2 + 1/n^2  = (i+1)/n^2$ as desired. The second line follows from the IH.
\end{proof}


\noindent A nearly identical lemma is useful for upper-bounding the makespan of \stb. The main difference arises from addressing the decreasing window sizes in a run, and this necessitates the condition that $w_i \geq m_i + \sqrt{m_i}$ rather than $w_i \geq m_0 + \sqrt{m_0}$ for all $i\geq 0$.  Later in Section~\ref{sec:reanalyze}, we start analyzing \stb when the window size reaches $4n$; this motivates the condition that $w_i \geq 4n/2^i$ our next lemma.

\begin{lemma}\label{lem:induction-new-stb}
Consider a batch of $m_0\leq n$ packets that execute over windows of size $w_i \geq m_i + \sqrt{m_i}$ and $w_i \geq 4n/2^i$ for all $i\geq 0$. If $m_i \geq n^{7/10}$, then $m_{i+1} \leq (4/5)\frac{m_0}{2^i 2^{2^i\lg(5/4)}}$ with error probability at most $(i+1)/n^2$.
\end{lemma}
\begin{proof}
We argue  by induction on $i\geq 0$.
\medskip

\noindent{\bf{Base Case}.} Nearly identical to the base case in proof of Lemma~\ref{lem:induction-new}; note the bound on $m_{i+1}$ is identical for $i=0$.\medskip

\noindent{\bf{Induction Hypothesis (IH)}.} For  $i\geq 1$,  assume $m_{i} \leq (4/5)\frac{m_0}{2^{i-1} 2^{2^{i-1}\lg(5/4)}}$ with error probability at most $i/n^2$.
\medskip

\noindent{\bf{Induction Step}.} For window $i\geq 1$, we wish to show that $m_{i+1} \leq (4/5)\frac{m_0}{2^i2^{2^i\lg(5/4)}}$ with an error bound of $(i+1)/n^2$ (we use the same $\epsilon$ as in Lemma~\ref{lem:induction-new}). Addressing the number of packets, we have:
\begin{eqnarray*}
m_{i+1}&\leq& \frac{(5/4)m_i^2}{w_i} \\ 
&\leq&\left(\frac{4 m_0}{5\cdot 2^{i-1} 2^{2^{i-1}\lg(5/4)}}\right)^2\left(\frac{5}{4w_i} \right)\\
&\leq&   \left(\frac{4m_0}{5\cdot 2^i 2^{2^{i}\lg(5/4)}}\right)\left(\frac{m_0}{2^{i-2}w_i} \right) \\
& \leq &   \left(\frac{4  m_0}{5\cdot 2^i 2^{2^{i}\lg(5/4)}}\right) \mbox{since $w_i \geq 4n/2^i$}
 \end{eqnarray*}
\noindent Again, the first line follows from Case 1 of Corollary~\ref{cor:remaining}, which we may invoke since $w_i \geq m_0 + \sqrt{m_0}$  for all $i\geq 0$, and $m_i \geq n^{7/10}$ by assumption. This gives the desired error bound of $i/n^2 + 1/n^2  = (i+1)/n^2$. The second line follows from the IH.
\end{proof}

\section{Bounding Makespan}\label{sec:reanalyze}

In this section, we pull together the results established earlier in order to bound the makespan of several backoff algorithms.

\subsection{Describing the Backoff Algorithms}\label{sec:describe-backoff}
We begin by describing the windowed backoff algorithms \FB (\fb), \BEB (\beb), and \LLB (\llb) analyzed in~\cite{BenderFaHe05}. Recall that, in each window, a packet selects a single slot uniformly at random to send in. Therefore, we need only specify how the size of successive windows change.

\fb is the simplest, with  all windows having size $\Theta(n)$. The value of hidden constant does not appear to be explicitly specified in the literature, but we observe that Bender~et al.~\cite{BenderFaHe05} use $3e^3$ in their upper-bound analysis. Here, we succeed using a smaller constant; namely, any value at least $1+1/\sqrt{n}$.

We let \beb use an initial window size of $1$, although any constant will yield the same asymptotic behavior. Under \beb, each successive window doubles in size.  

We let \llb use an initial window size of $4$; again, any constant will yield the same asymptotic behavior. For a current window size of $w_i$, it executes $\lceil \lg\lg(w_i) \rceil$ windows of that size before doubling; we call such a sequence of same-sized windows a \defn{plateau}.\footnote{As stated by Bender et al.~\cite{BenderFaHe05}, an equivalent specification (in terms of asymptotic bounds on makespan) of \llb is that $w_{i+1} = (1 + 1/\lg\lg(w_i))w_i$. This equivalence is elaborated on by Anderton et al.~\cite{anderton2021windowed}. } 

STB is non-monotonic and executes over a doubly-nested loop. The outer loop sets the current window size $w$ to be double that used in the preceding outer loop and each packet selects a single slot to send in; this is like BEB. Additionally, for each such $w$, the inner loop executes over $\lg w$ windows of decreasing size: $w, w/2, w/4, ..., 1$;  this sequence of windows is referred to as a \defn{run}. For each window in a run, a packet chooses a slot uniformly at random in which to send.


\subsection{Analysis}\label{sec:final-analysis}
The following results employ tools from the prior sections. In the following arguments, it is easy to check that the number of times these tools are invoked yields a total error probability which is at most $O(1/n)$. Thus, our theorems on makespan below hold with probability at least $1-O(1/n)$, and we omit further discussion of error.


\begin{theorem}\label{thm:FB}
The makespan of \fb with window size at least $n+\sqrt{n}$ is at most $n\lg\lg n +  O(n)$. 
\end{theorem}
\begin{proof}
Since $w_i \geq n+\sqrt{n}$ for all $i\geq 0$, by Lemma~\ref{lem:induction-new}  less than $n^{7/10}$ packets remain after $\lg\lg(n) + 1$  windows;  to see this, solve for $i$ in $(4/5)\frac{n}{2^{2^i\lg(5/4)}} = n^{7/10}$. By Lemma~\ref{lem:end}, all remaining packets succeed within $6$ more windows. The corresponding number of slots is $(\lg\lg n + 7) (n + \sqrt{n}) = n\lg\lg n + O(n)$. 
\end{proof}

\begin{theorem}\label{thm:BEB-upper}
The makespan of \beb is at most $512n\lg n + O(n)$.
\end{theorem}
\begin{proof}
Let $W$ be the first window of size at least $n+\sqrt{n}$ (and less than $2(n+\sqrt{n})$).  Assume no packets finish before the start of $W$; otherwise, this can only improve the makespan. By Lemma~\ref{lem:induction-new}  less than $n^{7/10}$ packets remain after $ \lg\lg(n) + 1$  windows. By Lemma~\ref{lem:end} all remaining packets succeed within $6$ more windows. Since $W$ has size less than  $2(n+\sqrt{n})$, the number of slots until the end of $W$, plus those for the $\lg\lg(n)+7$ subsequent windows, is less than:
\begin{eqnarray*}
 &&\left(\sum_{j=0}^{\lg(2(n+\sqrt{n}))}2^j\right) + \left(\sum_{k=1}^{\lg\lg(n) + 7} 2(n+\sqrt{n})2^k \right)\\
  &=& 512 (n+\sqrt{n})\lg n + O(n) 
\end{eqnarray*}
\noindent  by the sum of a geometric series.
\end{proof}

\begin{theorem}\label{thm:STB}
The makespan of \stb is $O(n)$.
\end{theorem}
\begin{proof}
Let $W$ be the first window of size at least $4n$.  Assume no packets finish before the start of $W$, that is $m_0=n$; else, this can only improve the makespan. 

While $m_i \geq n^{7/10}$, our analysis examines the windows in the run starting with window $W$, and so $w_0\geq 4n, w_1\geq 2n$, etc.  To invoke Lemma~\ref{lem:induction-new-stb}, we must ensure that the condition $w_i \geq m_i +\sqrt{m_i}$ holds in each window of this run.  This holds for $i=0$, since $w_0 = 4n \geq n+\sqrt{n}$. 

For $i\geq 1$,  we argue this inductively by proving $m_i \leq (5/4)^{2^{i-1}-1}\frac{n}{3^{2^{i-1}}}$. For the base case $i=1$, Lemma~\ref{lem:concentration} implies that $m_1 \leq n(1 - e^{-n/(4n-1)}  +\epsilon) \leq n(1 - e^{-1/3} + \epsilon) \leq n/3$, where $\epsilon$ is given in Lemma 6. For the inductive step, assume that $m_i \leq (5/4)^{2^{i-1}-1}\frac{n}{3^{2^{i-1}}}$ for all $i\geq 2$. Then, by Case 1 of Corollary~\ref{cor:remaining}:
\begin{eqnarray*}
m_{i+1} &\leq& (5/4)m_i^2/n \\
& \leq & (5/4) \left((5/4)^{2^{i-1}-1}\frac{n}{3^{2^{i-1}}}\right)^2/n\\
& \leq & (5/4)^{2^{i}-1}\frac{n}{3^{2^{i}}}
 \end{eqnarray*}
\noindent where the second line follows from the assumption, and so the inductive step holds.  On the other hand, at window $i$, $w_i\geq \frac{4n}{2^i} > \frac{4n}{(5/2)\cdot (12/5)^{2^{i-1}}} = 2\cdot (5/4)^{2^{i-1}-1}\frac{n}{3^{2^{i-1}}}  \geq 2 m_i > m_i + \sqrt{m_i}$ holds.

Lemma~\ref{lem:induction-new-stb} implies that after $ \lg\lg n + O(1)$  windows in this run, less than $n^{7/10}$ packets remain. Pessimistically, assume no other packets finish in the run. The next run starts with a window of size at least $8n$, and by
Lemma~\ref{lem:end}, all remaining packets succeed within the first $6$ windows of this run.

We have shown that \stb terminates within at most $\lceil \lg(n) \rceil + O(1)$ runs. The total number of slots over all of these runs is $O(n)$ by a geometric series.
\end{proof}

It is worth noting that \stb has  asymptotically-optimal makespan since we cannot hope to finish $n$ packets in $o(n)$ slots.

Bender et al.~\cite{BenderFaHe05} show that the optimal makespan for any {\it monotonic} windowed backoff algorithm is $O(n\lg\lg n/\lg\lg\lg n)$ and that \llb achieves this.  We re-derive the makespan for \llb.

\begin{figure}[t]
\hspace{0cm}\includegraphics[width=1.01\textwidth, trim = 0in 0in 0in 0in, clip]{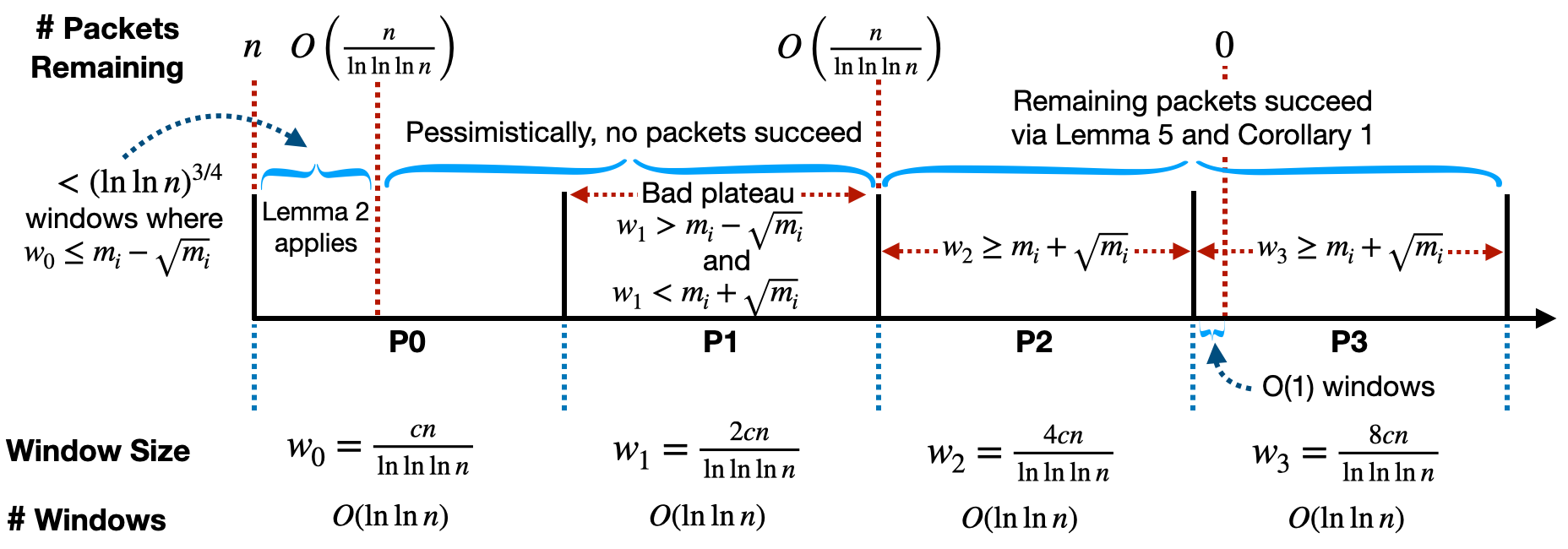}  
\caption{Illustration of the argument used to prove the makespan for \llb.}\label{f:llb} 
\end{figure}

\begin{theorem}\label{thm:LLB}
The makespan of \llb is $O\left(\frac{n\lg\lg n}{\lg\lg\lg n}\right)$.
\end{theorem}
\begin{proof}
We consider four consecutive plateaus in our argument, denoted by P0, P1, P2, and P3. Recall that, for any fixed plateau, its windows are the same size; therefore, we denote the window sizes for P0, P1, P2, and P3 by $w_0$, $w_1$, $w_2$, and $w_3$, respectively. Within each plateau, we use $m_i$ to denote the number of packets at the start of $i^{th}$ window of that plateau, for $i\geq 0$.

For ease of presentation, we break our analysis into two parts. The first deals with P0, where we address windows of size at most $m_i - \sqrt{m_i}$, allowing us to leverage Lemma~\ref{lem:concentration}. The second addresses P1, P2, and P3, where we show an eventual transition to windows of size at least $m_i + \sqrt{m_i}$, again allowing us to invoke Lemma~\ref{lem:concentration} and finish the remaining packets.  Figure~\ref{f:llb} highlights the structure of our argument.

\medskip

\noindent{\it \underline{Part 1}.} Le P0 be the first plateau with window size $w_0=cn/\ln\ln\ln n$ for some constant $c\geq 4$. Pessimistically, assume that no packets succeed before P0; otherwise, this can only improve the makespan.  

In P0, consider those windows $i\geq 0$ where $m_i$ satisfies $w_0 \leq m_i - \sqrt{m_i}$. By Lemma~\ref{lem:concentration}, each such window finishes at least the following number of packets:
\begin{eqnarray*}
\frac{(1-\epsilon)m_i}{e^{\frac{m_i}{(cn/\ln\ln\ln n) - 1}}}  &>& \frac{(1-\epsilon)n}{e^{\frac{n}{(cn/\ln\ln\ln n) - 1}}\cdot \ln\ln\ln n} \\
&\geq & \frac{(1-\epsilon)n}{(\ln\ln n)^{\frac{2}{c}} \cdot\ln\ln\ln n}  \\
&=& \frac{(1-\epsilon)n}{(\ln\ln n)^{\frac{\ln\ln\ln\ln n}{\ln\ln\ln n}+ \frac{2}{c}}}\\
&>& \frac{n}{(\ln\ln n)^{ \frac{3}{c}}} 
\end{eqnarray*}
\noindent where the first inequality follows noting that: $m_i \leq n$, which is placed in the exponent of the denominator); and since $w_0 \leq m_i - \sqrt{m_i}$, we have $m_i > n/\ln\ln\ln n$ in the windows considered, and this value is used in the numerator. The third line holds since  $(\ln\ln n)^{\ln (\ln\ln\ln n)} = (\ln\ln\ln n)^{\ln (\ln\ln n)}$, and the last line follows for sufficiently large $n$. It suffices to set: 
$$\epsilon = \sqrt{\frac{4e\ln^2(n)}{n}}$$

\noindent in order to obtain an error probability at most:
 $$\exp\left(-\frac{4e\ln^2(n)}{n} \cdot \frac{n}{2\ln\ln\ln(n) e^{\frac{n}{(cn/\ln\ln\ln n) - 1} } } \right) =  O\left(\frac{1}{n^2}\right).$$

It follows that after $\frac{n-n/\ln\ln\ln n}{n/(\ln\ln n)^{3/c}} < (\ln\ln n)^{3/4}$ windows in P0, the number of remaining packets is $O(n/\ln\ln\ln n)$. Also, for any window with index $i \geq (\ln\ln n)^{3/4}$ in P0, we have $w_0 > m_i - \sqrt{m_i}$; note that this also implies $w_j > m_i - \sqrt{m_i}$  for $j=1,2,3$ and for any $i$, given that windows in subsequent plateaus are larger, and the number of packets can only decrease. \medskip

\noindent{\it \underline{Part 2}.}  Most packets succeed in Part 1, and now we would now like to show that those remaining finish quickly. However, as established above, at some point in P0, we have  $w_j \geq  m_i - \sqrt{m_i} + 1$ for  $j=1,2,3$ and for any $i$. Therefore, we must wait for windows of size at least $m_i + \sqrt{m_i}$, since this is the alternative condition for invoking  Lemma~\ref{lem:concentration}, which is crucial to our later lemmas for showing packets succeed.

Fortunately, we should not expect to wait too long, since the window size continues to increase, while the number of packets can only remain the same or decrease. In the worst case,  P1 contains windows of size $w_1 = m_i - \sqrt{m_i} + 1$.  Then, at the start of P2, the window size becomes $2w_i = 2m_i - 2\sqrt{m_i}+2 \geq m_i +\sqrt{m_i}$ for $m_i \geq 4$. Therefore, we have at most one such ``bad'' plateau with P1, where pessimistically no packets succeed, after which both $w_2$ and $w_3$ in P2 and P3, respectively, are at least $m_i + \sqrt{m_i}$ for all $i$, so long as $m_i\geq 9$.



Starting with P2 and treating the remaining $O(n/\ln\ln\ln n)$ packets as an ``initial batch'', we invoke  Lemma~\ref{lem:induction-new}, which implies that after at most $\lg\lg(n)+1$ windows, less than $n^{7/10}$ packets will remain. Does P2 have enough windows? The number of windows in P2 is $\lg\lg(4cn/\ln\ln\ln n)$, and we note the upper and lower bounds:
$$\lg\lg(n) - 1 < \lg\lg\left(\frac{4cn}{\ln\ln\ln n}\right) < \lg\lg(n) + 1.$$
Therefore, we must proceed into P3 to reduce the number of packets to less than $n^{7/10}$. Then, if {\it at least} $n^{2/5}$ packets still remain, by Case 2 of Corollary 1, at most $O(n^{2/5})$ packets remain by the end of the next window, and they will finish within an additional $6$ windows by Lemma ~\ref{lem:end}. 

Finally, tallying up over the slots in $P0$ through $P3$, the makespan is $O(\ln\ln n)\times\, \allowbreak O(\frac{n}{\ln\ln\ln n})=O(\frac{n\ln\ln n}{\ln\ln\ln n})$.
\end{proof}




%

\section{Summary and Future Work}\label{sec:conclusion}

We have argued that standard Chernoff bounds can be applied to analyze singletons, and we illustrate how they simplify the analysis of several backoff algorithms under a batch of packets. We believe that the results presented here demonstrate the benefits of this approach, but there are some obvious extensions that may be of interest. First, we believe that lower bounds on the makespan of these backoff algorithms can be proved using this approach. Second,  a similar treatment is likely possible for polynomial backoff or generalized exponential backoff (see~\cite{BenderFaHe05} for the specification of these algorithms).  Third, it may be interesting to examine whether this analysis can be extended  to the case where packets have different sizes, as examined in~\cite{bender:heterogeneous}.

\medskip\medskip

\noindent{\bf Acknowledgements.} We are grateful to the anonymous reviewers; their comments significantly improved our manuscript.



\begin{thebibliography}{10}
\expandafter\ifx\csname url\endcsname\relax
  \def\url#1{\texttt{#1}}\fi
\expandafter\ifx\csname urlprefix\endcsname\relax\def\urlprefix{URL }\fi
\expandafter\ifx\csname href\endcsname\relax
  \def\href#1#2{#2} \def\path#1{#1}\fi

\bibitem{802.11-standard}
{IEEE} standard for information technology--telecommunications and information
  exchange between systems local and metropolitan area networks -- {Specific}
  requirements - {Part} 11: Wireless {LAN} medium access control ({MAC}) and
  physical layer ({PHY}) specifications, {IEEE} Std 802.11-2016 (Revision of
  {IEEE} Std 802.11-2012) (2016) 1--3534.

\bibitem{sun:backoff}
X.~Sun, L.~Dai, {Backoff design for IEEE 802.11 DCF networks: Fundamental
  tradeoff and design criterion}, IEEE/ACM Transactions on Networking 23~(1)
  (2015) 300--316.

\bibitem{kurose:computer}
J.~F. Kurose, K.~Ross, Computer Networking: A Top-Down Approach, 6th Edition,
  Pearson, 2013.

\bibitem{BenderFaHe05}
M.~A. Bender, M.~Farach-Colton, S.~He, B.~C. Kuszmaul, C.~E. Leiserson,
  Adversarial contention resolution for simple channels, in: Proceedings of the
  17th Annual ACM Symposium on Parallelism in Algorithms and Architectures
  (SPAA), 2005, pp. 325--332.

\bibitem{MetcalfeBo76}
R.~M. Metcalfe, D.~R. Boggs, Ethernet: Distributed packet switching for local
  computer networks, Communications of the ACM 19~(7) (1976) 395--404.

\bibitem{7282923}
D.~{Yin}, K.~{Lee}, R.~{Pedarsani}, K.~{Ramchandran}, {Fast and robust
  compressive phase retrieval with sparse-graph codes}, in: 2015 IEEE
  International Symposium on Information Theory (ISIT), 2015, pp. 2583--2587.

\bibitem{Karlin:1986:PHE:12130.12146}
A.~R. Karlin, E.~Upfal, {Parallel hashing - An efficient Implementation of
  shared memory}, in: Proceedings of the $18^{th}$ Annual ACM Symposium on
  Theory of Computing (STOC), 1986, pp. 160--168.

\bibitem{Upfal:1984:ESP:828.1892}
E.~Upfal, {Efficient schemes for parallel communication}, Journal of the ACM
  31~(3) (1984) 507--517.

\bibitem{Mitzenmacher:1996:PTC:924815}
M.~D. Mitzenmacher, The power of two choices in randomized load balancing,
  Ph.D. thesis (1996).

\bibitem{Mitzenmacher:2005:PCR:1076315}
M.~Mitzenmacher, E.~Upfal, {Probability and Computing: Randomized Algorithms
  and Probabilistic Analysis}, Cambridge University Press, New York, NY, USA,
  2005.

\bibitem{dubhashi:concentration}
D.~Dubhashi, A.~Panconesi, Concentration of Measure for the Analysis of
  Randomized Algorithms, 1st Edition, Cambridge University Press, 2009.

\bibitem{dubhashi:negative}
D.~Dubhashi, D.~Ranjan, {Balls and bins: A study in negative dependence},
  Random Structures \& Algorithms 13~(2) (1998) 99--124.

\bibitem{GreenbergL85}
R.~I. Greenberg, C.~E. Leiserson, Randomized routing on fat-trees, in:
  Proceedings of the $26^{th}$ Annual Symposium on Foundations of Computer
  Science (FOCS), 1985, pp. 241--249.

\bibitem{Gereb-GrausT92}
M.~Ger{\'{e}}b{-}Graus, T.~Tsantilas, Efficient optical communication in
  parallel computers, in: Proceedings $4^{th}$ Annual ACM Symposium on Parallel
  Algorithms and Architectures (SPAA), 1992, pp. 41--48.

\bibitem{AzarBrKa99}
Y.~Azar, A.~Z. Broder, A.~R. Karlin, E.~Upfal, Balanced allocations, SIAM
  Journal on Computing 29~(1) (1999) 180--200.

\bibitem{ColeFrMa98}
R.~Cole, A.~M. Frieze, B.~M. Maggs, M.~Mitzenmacher, A.~W. Richa, R.~K.
  Sitaraman, E.~Upfal, On balls and bins with deletions, in: Proceedings of the
  $2^{nd}$ International Workshop on Randomization and Approximation Techniques
  in Computer Science (RANDOM), 1998, pp. 145--158.

\bibitem{RichaMi01}
A.~W. Richa, M.~Mitzenmacher, R.~Sitaraman, The power of two random choices: A
  survey of techniques and results, Combinatorial Optimization 9 (2001)
  255--304.

\bibitem{Vocking03}
B.~{V\"{o}cking}, How asymmetry helps load balancing, Journal of the ACM 50~(4)
  (2003) 568--589.

\bibitem{BerenbrinkCzEg12}
P.~Berenbrink, A.~Czumaj, M.~Englert, T.~Friedetzky, L.~Nagel, Multiple-choice
  balanced allocation in (almost) parallel, in: Approximation, Randomization,
  and Combinatorial Optimization. Algorithms and Techniques (APPROX-RANDOM),
  2012, pp. 411--422.

\bibitem{BerenbrinkCzSt06}
P.~Berenbrink, A.~Czumaj, A.~Steger, B.~V{\"{o}}cking, Balanced allocations:
  The heavily loaded case, {SIAM} Journal on Computing 35~(6) (2006)
  1350--1385.

\bibitem{BerenbrinkKhSa13}
P.~Berenbrink, K.~Khodamoradi, T.~Sauerwald, A.~Stauffer, Balls-into-bins with
  nearly optimal load distribution, in: Proceedings of the $25^{th}$ Annual ACM
  Symposium on Parallelism in Algorithms and Architectures (SPAA), 2013, pp.
  326--335.

\bibitem{lenzen:tight}
C.~Lenzen, R.~Wattenhofer, {Tight bounds for parallel randomized load
  balancing}, in: Proceedings of the $43^{rd}$ Annual ACM Symposium on Theory
  of Computing (STOC), 2011, pp. 11--20.

\bibitem{czumaj:randomized}
A.~{Czumaj}, V.~{Stemann}, {Randomized allocation processes}, in: Proceedings
  $38^{th}$ Annual Symposium on Foundations of Computer Science (FOCS), 1997,
  pp. 194--203.

\bibitem{hastad:analysis}
J.~Hastad, T.~Leighton, B.~Rogoff, Analysis of backoff protocols for multiple
  access channels, SIAM Journal on Computing 25~(4) (740-774) 1996.

\bibitem{GoodmanGrMaMa88}
J.~Goodman, A.~G. Greenberg, N.~Madras, P.~March, Stability of binary
  exponential backoff, Journal of the ACM 35~(3) (1988) 579--602.

\bibitem{RaghavanUp99}
P.~Raghavan, E.~Upfal, Stochastic contention resolution with short delays, SIAM
  Journal on Computing 28~(2) (1999) 709--719.

\bibitem{GOLDBERG1999232}
Analysis of practical backoff protocols for contention resolution with multiple
  servers, Journal of Computer and System Sciences 58~(1) (1999) 232 -- 258.

\bibitem{goldberg:contention}
L.~A. Goldberg, P.~D. Mackenzie, M.~Paterson, A.~Srinivasan, Contention
  resolution with constant expected delay, Journal of the ACM 47~(6) (2000)
  1048--1096.

\bibitem{bender:heterogeneous}
M.~A. Bender, J.~T. Fineman, S.~Gilbert, {Contention Resolution with
  heterogeneous job sizes}, in: Proceedings of the $14^{th}$ Conference on
  Annual European Symposium (ESA), 2006, pp. 112--123.

\bibitem{FernandezAntaMoMu13}
A.~F. Anta, M.~A. Mosteiro, J.~R. Mu{\~{n}}oz, Unbounded contention resolution
  in multiple-access channels, Algorithmica 67~(3) (2013) 295--314.

\bibitem{willard:loglog}
D.~E. Willard, Log-logarithmic selection resolution protocols in a multiple
  access channel, SIAM Journal on Computing 15~(2) (1986) 468--477.

\bibitem{nakano2002}
K.~Nakano, S.~Olariu, Uniform leader election protocols for radio networks,
  IEEE Transactions on Parallel and Distributed Systems 13~(5) (2002) 516--526.

\bibitem{FinemanNW16}
J.~T. Fineman, C.~Newport, T.~Wang, Contention resolution on multiple channels
  with collision detection, in: Proceedings of the {ACM} Symposium on
  Principles of Distributed Computing (PODC), 2016, pp. 175--184.

\bibitem{fineman:contention}
J.~T. Fineman, S.~Gilbert, F.~Kuhn, C.~Newport, Contention resolution on a
  fading channel, in: Proceedings of the ACM Symposium on Principles of
  Distributed Computing (PODC), 2016, pp. 155--164.

\bibitem{ChangKPWZ17}
Y.~Chang, T.~Kopelowitz, S.~Pettie, R.~Wang, W.~Zhan, Exponential separations
  in the energy complexity of leader election, in: Proceedings of the $49^{th}$
  Annual {ACM} Symposium on Theory of Computing (STOC), 2017, pp. 771--783.

\bibitem{Chang:2018:ECB:3212734.3212774}
Y.-J. Chang, V.~Dani, T.~P. Hayes, Q.~He, W.~Li, S.~Pettie, The energy
  complexity of broadcast, in: Proceedings of the 2018 ACM Symposium on
  Principles of Distributed Computing (PODC), 2018, pp. 95--104.

\bibitem{doi:10.1137/140982763}
G.~D. Marco, D.~R. Kowalski, Fast nonadaptive deterministic algorithm for
  conflict resolution in a dynamic multiple-access channel, SIAM Journal on
  Computing 44~(3) (2015) 868--888.

\bibitem{DeMarco:2017:ASC:3087801.3087831}
G.~De~Marco, G.~Stachowiak, Asynchronous shared channel, in: Proceedings of the
  ACM Symposium on Principles of Distributed Computing (PODC), 2017, pp.
  391--400.

\bibitem{Bender:2016:CRL:2897518.2897655}
M.~A. Bender, T.~Kopelowitz, S.~Pettie, M.~Young, Contention resolution with
  log-logstar channel accesses, in: Proceedings of the $48^{th}$ Annual ACM
  Symposium on Theory of Computing (STOC), 2016, pp. 499--508.

\bibitem{DBLP:conf/spaa/AgrawalBFGY20}
K.~Agrawal, M.~A. Bender, J.~T. Fineman, S.~Gilbert, M.~Young, Contention
  resolution with message deadlines, in: Proceedings of the $32^{nd}$ {ACM}
  Symposium on Parallelism in Algorithms and Architectures (SPAA), 2020, pp.
  23--35.

\bibitem{chlebus:better}
B.~S. Chlebus, D.~R. Kowalski, A better wake-up in radio networks, in:
  Proceedings of $23^{rd}$ ACM Symposium on Principles of Distributed Computing
  (PODC), 2004, pp. 266--274.

\bibitem{chlebus:wakeup}
B.~S. Chlebus, L.~Gasieniec, D.~R. Kowalski, T.~Radzik, On the wake-up problem
  in radio networks, in: Proceedings of the $32^{nd}$ International Colloquium
  on Automata, Languages and Programming (ICALP), 2005, pp. 347--359.

\bibitem{chrobak:wakeup}
M.~Chrobak, L.~Gasieniec, D.~R. Kowalski, The wake-up problem in multihop radio
  networks, SIAM Journal on Computing 36~(5) (2007) 1453--1471.

\bibitem{Chlebus:2016:SWM:2882263.2882514}
B.~S. Chlebus, G.~De~Marco, D.~R. Kowalski, Scalable wake-up of multi-channel
  single-hop radio networks, Theoretical Computer Science 615 (2016) 23 -- 44.

\bibitem{DEMARCO20171}
G.~D. Marco, D.~R. Kowalski, Contention resolution in a non-synchronized
  multiple access channel, Theoretical Computer Science 689 (2017) 1 -- 13.

\bibitem{Jurdzinski:2015:CSM:2767386.2767439}
T.~Jurdzinski, G.~Stachowiak, The cost of synchronizing multiple-access
  channels, in: Proceedings of the ACM Symposium on Principles of Distributed
  Computing (PODC), 2015, pp. 421--430.

\bibitem{awerbuch:jamming}
B.~Awerbuch, A.~Richa, C.~Scheideler, A jamming-resistant \textsc{MAC} protocol
  for single-hop wireless networks, in: Proceedings of the $27^{th}$ ACM
  Symposium on Principles of Distributed Computing (PODC), 2008, pp. 45--54.

\bibitem{richa:jamming2}
A.~Richa, C.~Scheideler, S.~Schmid, J.~Zhang, A jamming-resistant {MAC}
  protocol for multi-hop wireless networks, in: Proceedings of the
  International Symposium on Distributed Computing (DISC), 2010, pp. 179--193.

\bibitem{richa:jamming3}
A.~Richa, C.~Scheideler, S.~Schmid, J.~Zhang, Competitive and fair medium
  access despite reactive jamming, in: Proceedings of the $31^{st}$
  International Conference on Distributed Computing Systems (ICDCS), 2011, pp.
  507--516.

\bibitem{richa:jamming4}
A.~Richa, C.~Scheideler, S.~Schmid, J.~Zhang, Competitive and fair throughput
  for co-existing networks under adversarial interference, in: Proceedings of
  the ACM Symposium on Principles of Distributed Computing (PODC), 2012, pp.
  291--300.

\bibitem{ogierman:competitive}
A.~Ogierman, A.~Richa, C.~Scheideler, S.~Schmid, J.~Zhang, {Sade: Competitive
  {MAC} under adversarial {SINR}}, Distributed Computing 31~(3) (2018)
  241--254.

\bibitem{Anantharamu2011}
L.~Anantharamu, B.~S. Chlebus, D.~R. Kowalski, M.~A. Rokicki, Medium access
  control for adversarial channels with jamming, in: Proceedings of the
  $18^{th}$ International Colloquium on Structural Information and
  Communication Complexity (SIROCCO), 2011, pp. 89--100.

\bibitem{BenderFiGi16}
M.~A. Bender, J.~T. Fineman, S.~Gilbert, M.~Young, How to scale exponential
  backoff: {Constant} throughput, polylog access attempts, and robustness, in:
  Proceedings of the $27^{th}$ Annual ACM-SIAM Symposium on Discrete Algorithms
  (SODA), 2016, pp. 636--654.

\bibitem{DBLP:journals/jacm/BenderFGY19}
M.~A. Bender, J.~T. Fineman, S.~Gilbert, M.~Young, Scaling exponential backoff:
  Constant throughput, polylogarithmic channel-access attempts, and robustness,
  Journal of the {ACM} 66~(1) (2019) 6:1--6:33.

\bibitem{anderton2021windowed}
W.~C. Anderton, T.~Chakraborty, M.~Young, Windowed backoff algorithms for
  {WiFi}: {Theory} and performance under batched arrivals, Distributed
  Computing 34~(5) (2021) 367--393.

\end{thebibliography}

\clearpage
\appendix

\noindent{\bf\LARGE Appendix}

\section{Chernoff Bounds and Property 1}\label{app:property1}

In the derivation of  Chernoff bounds  given in Dubhashi and Panconesi~\cite{dubhashi:concentration}, the equality version for Equation~\ref{equ:mgf} is claimed in the line above Equation 1.3 on page $4$ by invoking independence of the random variables. However, the indicator variables for counting singletons are dependent. We can show, in Claim~\ref{claim-chernoff} below,  that Property~\ref{prop1} leads to Equation~\ref{equ:mgf}, where the direction of the inequality aligns with the derivation of Chernoff bounds.



\begin{claim}\label{claim-chernoff}
Let $X_1,\cdots,X_n$ be a set of indicator random variables satisfying the property:
\begin{equation}\label{equ:cond}
Pr\[\bigwedge\limits _{i \in \mathbbm{S}} \{X_i=1\} \right] \leq \prod_{i \in \mathbbm{S}} Pr\[X_i=1\]
\end{equation}
for all subsets $\mathbbm{S} \subset\{1,\cdots,n\}$. Then the following holds:
\begin{equation}\label{equ:mgf}
E\[\prod_{i=1}^n e^{\lambda X_i}\] \leq \prod_{i=1}^n E\[e^{\lambda X_i}\]
\end{equation}
\end{claim}

\begin{proof}
We will show that the both sides of equation (\ref{equ:mgf}) can be expressed as polynomial functions of $\lambda$ with a similar form. First, let us look at the left hand side (LHS). By the Taylor expansion $e^{\lambda X_i} = \sum_{k=0}^{\infty} \lambda^k\frac{X_i^k}{k!}$,  the product $\prod_{i=1}^n e^{\lambda X_i} = \prod_{i=1}^n \left(\sum_{k=0}^{\infty} \frac{\lambda^k}{k!}X_i^k\right)$, which can be expressed as $\sum_{r=0}^{\infty} f_r \lambda^r$ with $f_0=1$ and the other coefficients $f_r$, $r\geq 1$, being functions of $X_i$'s. The expression of $f_r$, $r\geq 1$, is given as $f_r=\sum_{k=1}^{\min\{r,n\}}\sum_{(i_1,i_2,\cdots,i_k)\subset \{1,\cdots,n\}} f_r(i_1,\cdots,i_k)$ with 
\begin{equation}\label{equ:fr-k}
f_r(i_1,\cdots,i_k) = \sum_{\substack{(d_1,\cdots,d_k): d_1\leq d_2 \leq \cdots \leq d_k\\ d_1+d_2+\cdots + d_k=r}}  \frac{X_{i_1}^{d_1}}{d_1!} \frac{X_{i_2}^{d_2}}{d_2!} \cdots \frac{X_{i_k}^{d_k}}{d_k!},
\end{equation}
where $d_1,\cdots,d_k$ are positive integers. Equation (\ref{equ:fr-k}) tells that given a subset of size $k$, $\{X_{i_1},\cdots,X_{i_k}\}$, from $\{X_1,\cdots,X_n\}$, we choose all possible sets of increasing positive integers $d_1,\cdots,d_k$ as the power for $X_{i_1},\cdots,X_{i_k}$ while keeping their sum equal to $r$. Here, we list the expressions of $f_1$, $f_2$, and $f_3$ as illustrative examples: $f_1 = \sum_{i=1}^n X_i$, $f_2 = \sum_{i=1}^n \frac{X_i^2}{2!} + \sum_{1\leq i_1\neq i_2 \leq n} X_{i_1}X_{i_2}$, and 
$f_3 = \sum_{i=1}^n \frac{X_i^3}{3!} + \sum_{1\leq i_1\neq i_2 \leq n} X_{i_1}\frac{X_{i_2}^2}{2!} + \sum_{1\leq i_1\neq i_2 \neq i_3 \leq n} X_{i_1}X_{i_2}X_{i_3}$. With the expression of $f_r$'s, the LHS becomes
$$
\text{LHS} = 1 + \sum_{r=1}^{\infty} \lambda^r \sum_{k=1}^{\min\{r,n\}} \sum_{(i_1,i_2,\cdots,i_k)\subset \{1,\cdots,n\}} E[f_r(i_1,\cdots,i_k)],
$$
where $E[f_r(i_1,\cdots,i_k)]=\sum_{\substack{(d_1,\cdots,d_k): d_1\leq d_2 \leq \cdots \leq d_k\\ d_1+d_2+\cdots + d_k=r}}   \frac{E\[X_{i_1}^{d_1}X_{i_2}^{d_2}\cdots X_{i_k}^{d_k}\]}{d_1!d_2!\cdots d_k!}$.

We use similar derivations on the right hand side (RHS), which becomes $\text{RHS} = \prod_{i=1}^n \left(\sum_{k=0}^{\infty} \lambda^k \frac{E\[X_i^k\]}{k!}\right) = 1 + \displaystyle \sum_{r=1}^{\infty} \lambda^r \sum_{k=1}^{\min\{r,n\}}  \sum_{(i_1,i_2,\cdots,i_k)  \subset \{1,\cdots,n\}} \, \allowbreak  \tilde f_r(i_1,\cdots,i_k)$ with
$$
\tilde f_r(i_1,\cdots,i_k) = \sum_{\substack{(d_1,\cdots,d_k): d_1\leq d_2 \leq \cdots \leq d_k\\ d_1+d_2+\cdots + d_k=r}}  \frac{E[X_{i_1}^{d_1}]}{d_1!} \frac{E[X_{i_2}^{d_2}]}{d_2!} \cdots \frac{E[X_{i_k}^{d_k}]}{d_k!}.
$$
Note that $\tilde f_r(i_1,\cdots,i_k) $ is similar to $f_r(i_1,\cdots,i_k)$ in equation (\ref{equ:fr-k}) but replaces $X_{i_j}^{d_j}$, $j=1,\cdots,k$, with their expectations. Consequently, the relationship between the LHS and RHS lies solely on the difference between $E\left[X_{i_1}^{d_1}\cdots X_{i_k}^{d_k}\right]$ and $E\left[X_{i_1}^{d_1}\right]\cdots E\left[X_{i_k}^{d_k}\right]$. 

For indicator random variables, $X_i^k=X_i$ for any positive integer $k$. Thus, 
$$
E\left[X_{i_1}^{d_1}\cdots X_{i_k}^{d_k}\right]=E\left[X_{i_1}\cdots X_{i_k}\right]\, \text{and}\,  E\left[X_{i_1}^{d_1}\right]\cdots E\left[X_{i_k}^{d_k}\right]=E\left[X_{i_1}\right]\cdots E\left[X_{i_k}\right].
$$ 
In addition, $E[X_i]=Pr[X_i=1]$, so $E\left[X_{i_1}\right]\cdots E\left[X_{i_k}\right]=\prod_{j=\{i_1,\cdots,i_k\}}Pr[X_j=1]$. On the other hand, $E[X_{i_1}\cdots X_{i_k}]=Pr\left[\bigwedge_{j=\{i_1,\cdots,i_k\}}\{X_i=1\}\right]$.

Thus, if the property $Pr\[\bigwedge\limits _{i \in \mathbbm{S}} \{X_i=1\} \right] \leq \prod_{i \in \mathbbm{S}} Pr\[X_i=1\]$ is satisfied for all subsets $\mathbbm{S} \subset\{1,\cdots,n\}$, $
E\left[X_{i_1}^{d_1}\cdots X_{i_k}^{d_k}\right] \leq E\left[X_{i_1}^{d_1}\right]\cdots E\left[X_{i_k}^{d_k}\right]$ for all possible $(i_1,\cdots,i_k)\subset \{1,\cdots,n\}$. Consequently, LHS $\leq$ RHS.
\end{proof}
\end{document}